\documentclass{article}
\usepackage{amsmath,amssymb,amsthm,amsfonts,xcolor,graphicx,verbatim,dsfont,pdfpages,interval,pgfplots,tikz,authoraftertitle,enumitem,cleveref}
\newtheorem{theorem}{Theorem}[section]
\newtheorem{prop}[theorem]{Proposition}
\newtheorem{lemma}[theorem]{Lemma}

\newtheorem{defn}[theorem] {Definition}
\newtheorem{rem}[theorem] {Remark}

\newtheorem{note}[theorem] {Note}
\newtheorem{corollary}[theorem]{Corollary}
\newtheorem{Notation}[subsection]{Notation}

\newtheorem{Main result:}[theorem]  {Main result:}

\newlist{Properties}{enumerate}{10}
\setlist[Properties]{label*=\arabic*}

\begin{document}
\title{On a family of a linear maps from $M_{n}(\mathbb{C})$ to $M_{n^{2}}(\mathbb{C})$}
\author{Benoit Collins, Hiroyuki Osaka, Gunjan Sapra}
\maketitle
%%%%%%%%%%%%%%%%%%%%%%%%%%%%%%%%%%%%%%%%%%%%%%%%%%%%%%%%%%%%%%%%%%%%%%%%%%%%%%%%%%%%%%%%%%%%%%%%%%%%%%%%%%%%%%%%%%%%%%%%%%%%%%%%%%%%%%%%%%%%%%%%%%%%%%%%%%%%%%%%%%%%%%%%	
\begin{abstract}
Bhat \nocite{*}\cite{bhat2011} characterizes the family of linear maps defined on $B(\mathcal{H})$ which preserve unitary conjugation. We generalize this idea and study the maps with a similar equivariance property on finite-dimensional matrix algebras. We show that the maps with equivariance property are significant to study $k$-positivity of linear maps defined on finite-dimensional matrix algebras.\\
In \cite{choi1072positive} Choi showed that $n$-positivity is different from $(n-1)$-positivity for the linear maps defined on $n$ by $n$ matrix algebras.\\
In this paper, we present a parametric family of linear maps $\Phi_{\alpha, \beta,n} : M_{n}(\mathbb{C}) \rightarrow M_{n^{2}}(\mathbb{C})$ and study the properties of positivity, completely positivity, decomposability etc. We determine values of parameters $\alpha$ and $\beta$ for which the family of maps $\Phi_{\alpha, \beta,n}$ is positive for any natural number $n \geq 3$. We focus on the case of $n=3,$ that is, $\Phi_{\alpha, \beta,3}$ and study the properties of $2$-positivity, completely positivity and decomposability. In particular, we give values of parameters $\alpha$ and $\beta$ for which the family of maps $\Phi_{\alpha, \beta,3}$ is $2$-positive and not completely positive. \\
\vspace{2mm}
\end{abstract}

%%%%%%%%%%%%%%%%%%%%%%%%%%%%%%%%%%%%%%%%%%%%%%%%%%%%%%%%%%%%%%%%%%%%%%%%%%%%%%%%%%%%%%%%%%%%%%%%%%%%%%%%%%%%%%%%%%%%%%%%%%%%%%%%%%%%%%%%%%%%%%%%%%%%%%%%%%%%%%%%%%%%%%%%
		
\section{Introduction}
Recently, positive linear maps on matrix algebras have become important since they are widely used to detect entanglement in Quantum information theory \cite{Book1,Horodecki:1996nc}. 
Unfortunately, the structure of positive linear maps is not clearly understood. In this regard, it becomes important to study whether a positive linear map can be written as a sum of other simple maps, that is, whether a positive linear map is decomposable. Positive maps which can not be written in simple form (Indecomposable and atomic maps) are used to detect PPT entangled states in Quantum information theory\cite{STORMER20082303,Horodecki:1996nc}.
 Therefore, it is important to distinguish $k$-positivity from $(k-1)$-positivity of linear maps defined on finite-dimensional matrix algebras. Choi \cite{choi1072positive} gave an example of a linear map defined on $M_{n}(\mathbb{C})$ for $n \geq 3,$ which is $(n-1)$-positive but not completely positive. In this paper, we present a family of linear maps from $M_{n}(\mathbb{C})$ to   $M_{n^{2}}(\mathbb{C})$ to distinguish $k$-positivity from $(k-1)$-positivity for $2 < k \leq n$. First, we recall the basic definitions and theorems.\\ 
Let $\mathcal{A}$ be a $C^{*}$-algebra. Elements of the form $\{x^{*}x | x \in \mathcal{A}\}$ are called positive elements of $\mathcal{A}$. The set of all such elements is denoted by $\mathcal{A}^{+}.$
It is known that
\begin{equation*}
	\mathcal{A}^{+} = \{a \in \mathcal{A} |\hspace{1.5mm}  a=a^{*} \text{and} \hspace{1.5mm} Spec(a) \subseteq [0,\infty)\},
\end{equation*}
where $Spec(a)$ denotes the spectrum of $a \in \mathcal{A}.$\\
Throughout this paper, we focus on $\mathcal{A}= M_{d}(\mathbb{C})$, the space of all $d$ by $d$ matrices with complex entries.
%%%%%%%%%%%%%%%%%%%%%%%%%%%%%%%%%%%%%%%%%%%%%%%%%%%%%%%%%%%%%%%%%%%%%%%%%%%%%%%%%%%%%%%%%%%%%%%%%%%%%%%%%%%%%%%%%%%%%%%%%%%%%%%%%%%%%%%%%%%%%%%%%%%%%%%%%%%%%%%%%%%%%%%%
\begin{Notation}
$i_{d}$ denotes the identity map on $M_{d}(\mathbb{C})$ and $\mathds{1}_{d}$ denotes the identity matrix in $M_{d}(\mathbb{C})$.
$B(\mathcal{H})$ denotes the space of all bounded linear operators on a Hilbert space $\mathcal{H}$ and $B(M_{d_{1}}(\mathbb{C}),M_{d_{2}}(\mathbb{C}))$ denotes the space of all bounded linear maps from $M_{d_{1}}(\mathbb{C})$ to $M_{d_{2}}(\mathbb{C})$.\\
We denote $B_{d}= \sum_{i=1}^{d}|e_{i}\rangle \otimes |e_{i}\rangle$, the unnormalized maximally entangled $Bell$ vector in $\mathbb{C}^{d} \otimes \mathbb{C}^{d}$ and $B_{d^{2}}=B_{d}B_{d}^{*}.$	\\
$A^{t}$ and $\mathrm{Tr}(A)$ denote the transpose and trace of a matrix $A \in M_{d}(\mathbb{C})$ respectively. Let $T_{d}$ denote the transpose map on $M_{d}(\mathbb{C}) \hspace{1.5mm} i.e \hspace{1.5mm} T_{d}(A)=A^{t}$ for $A \in M_{d}(\mathbb{C})$.\\	
Let $(e_{ij})$ be the \emph{matrix units} of $M_{d}(\mathbb{C}),$ that is, $e_{ij}e_{kl}=\delta_{jk}e_{il}$ for $1 \leq i,j,k,l \leq d,\hspace{0.5mm} e_{ij}^{*}=e_{ji}$ and $\sum_{i=1}^{d}e_{ii}=\mathds{1}_{d}.$\\
\end{Notation}
%%%%%%%%%%%%%%%%%%%%%%%%%%%%%%%%%%%%%%%%%%%%%%%%%%%%%%%%%%%%%%%%%%%%%%%%%%%%%%%%%%%%
%%%%%%%%%%%%%%%%%%%%%%%%%%%%%%%%%%%%%%%%%%%%%%%%%%%%%%%%%%%%%%%%%%%%%%%%%%%%%%%%%%%%
\begin{defn} 
A matrix $A\in M_{d}(\mathbb{C})$ is called \emph{positive semi-definite} if it is hermitian and all its eigenvalues are positive. It is denoted as $A \geq 0.$ The set of all positive semi-definite matrices in $M_{d}(\mathbb{C})$ is denoted by $M_{d}(\mathbb{C})^{+}$.\\
A linear map $\Phi: M_{d_{1}}(\mathbb{C}) \rightarrow M_{d_{2}}(\mathbb{C})$ is called \emph{positive} if it maps positive elements of $M_{d_{1}}(\mathbb{C})$ to positive elements of $M_{d_{2}}(\mathbb{C}).$ It is called \emph{copositive} if $\Phi \circ T_{d_{1}}$ is positive. $\Phi$ is called \emph{$k$-positive} for $k \in \mathbb{N}$ if the linear  map $ i_{k} \otimes
\Phi : M_{k}(\mathbb{C}) \otimes M_{d_{1}}(\mathbb{C})  \rightarrow M_{k}(\mathbb{C}) \otimes M_{d_{2}}(\mathbb{C})$ is positive. It is called \emph{$k$-copositive} if $\Phi \circ T_{d_{1}}$ is $k$-positive.
\end{defn}

%%%%%%%%%%%%%%%%%%%%%%%%%%%%%%%%%%%%%%%%%%%%%%%%%%%%%%%%%%%%%%%%%%%%%%%%%%%%%%%%%%%%

\begin{defn}
A linear map $\Phi$ is called \emph{completely positive} (resp. \emph{completely copositive}) if $\Phi$ (resp. $\Phi \circ T_{d_{1}})$ is $k$-positive (resp. $k$-copositive) for every $k \in \mathbb{N}.$	
\end{defn}

%%%%%%%%%%%%%%%%%%%%%%%%%%%%%%%%%%%%%%%%%%%%%%%%%%%%%%%%%%%%%%%%%%%%%%%%%%%%%%%%%%%%

In the study of completely positive maps, the Choi matrix is very important. We define the Choi matrix of a linear map and study the relation between linear maps and its Choi matrix. 

%%%%%%%%%%%%%%%%%%%%%%%%%%%%%%%%%%%%%%%%%%%%%%%%%%%%%%%%%%%%%%%%%%%%%%%%%%%%%%%%%%%%
\begin{defn}
Let $\Phi: M_{d_{1}}(\mathbb{C}) \rightarrow M_{d_{2}}(\mathbb{C})$ be a linear map. Then the Choi matrix of $\Phi$ is defined as:
$$C_{\Phi} := \sum_{i,j=1}^{d_{1}}e_{ij} \otimes \Phi(e_{ij}). $$	
\end{defn}
%%%%%%%%%%%%%%%%%%%%%%%%%%%%%%%%%%%%%%%%%%%%%%%%%%%%%%%%%%%%%%%%%%%%%%%%%%%%%%%%%%%%
With the help of Choi matrix, we can define a correspondence between the spaces $B(M_{d_{1}}(\mathbb{C}), M_{d_{2}}(\mathbb{C}))$ and  $M_{d_{1}}(\mathbb{C})\otimes M_{d_{2}}(\mathbb{C})$.\\
Define $\Sigma:B(M_{d_{1}}(\mathbb{C}), M_{d_{2}}(\mathbb{C})) \rightarrow M_{d_{1}}(\mathbb{C})\otimes M_{d_{2}}(\mathbb{C})$ given by $$\Phi \mapsto C_{\Phi}.$$
This map is obviously linear and bijective. This isomorphism is known as the Choi-Jamiolkowski isomorphism \cite{JAMIOLKOWSKI1972275}.
		
\begin{theorem} \label{CPresult}\cite[Theorem 4.1.8]{størmer1963} 
Let $\Phi: M_{d_{1}}(\mathbb{C}) \rightarrow M_{d_{2}}(\mathbb{C})$ be a linear map. Then the following conditions are equivalent.
\begin{enumerate}
\item $\Phi$ is completely positive.
\item $C_{\Phi} \geq 0.$
\item $\Phi(A)=\sum_{i=1}^{k}V_{i}^{*}AV_{i}$ with $V_{i}:M_{d_{2}}(\mathbb{C})\rightarrow M_{d_{1}}(\mathbb{C})$, linear map and \\ $k \leq \min\{{d_{1}}{d_{2}}\}.$
\end{enumerate} 
\end{theorem}
%%%%%%%%%%%%%%%%%%%%%%%%%%%%%%%%%%%%%%%%%%%%%%%%%%%%%%%%%%%%%%%%%%%%%%%%%%%%%%%%%%%%
%%%%%%%%%%%%%%%%%%%%%%%%%%%%%%%%%%%%%%%%%%%%%%%%%%%%%%%%%%%%%%%%%%%%%%%%%%%%%%%%%%%%
\textbf{Main result:}
In this paper, we define a family of linear maps $\Phi_{\alpha,\beta,n}: M_{n}(\mathbb{C}) \rightarrow M_{n}(\mathbb{C}) \otimes M_{n}(\mathbb{C})$, where $n \geq 3$ is any natural number, $\alpha$ and $\beta$ are real numbers. We give values of parameters $\alpha$ and $\beta$ for which the family of maps $\Phi_{\alpha,\beta,3}$ is $2$-positive and not completely positive. Moreover, the values of parameters $\alpha$ and $\beta$ for which the family of maps $\Phi_{\alpha, \beta,n}$ is $(k-1)$-positive and not $k$-positive for natural numbers $n > 3$ and $2 < k \leq n$ can easily be deduced from Proposition \ref{k-positivity}.\\ 
This family of maps has some interesting equivariance properties. This equivariance property is significant to prove $k$-positivity of linear maps defined on finite-dimensional matrix algebras. The family of maps is defined as follows:
 	
	\begin{defn}
		Let $\alpha$ and $\beta$ be two real numbers and $n \geq 3.$ Then the family of maps, \\ $$\Phi_{\alpha, \beta,n}:M_{n}(\mathbb C) \rightarrow M_{n}(\mathbb C) \otimes M_{n}(\mathbb C) $$ defined by $$A \mapsto A^{t} \otimes \mathds{1}_{n}+ \mathds{1}_{n} \otimes A + \mathrm{Tr}(A)(\alpha \mathds{1}_{n^{2}}+\beta B_{n^{2}}).
		$$ 	\end{defn}
%%%%%%%%%%%%%%%%%%%%%%%%%%%%%%%%%%%%%%%%%%%%%%%%%%%%%%%%%%%%%%%%%%%%%%%%%%%%%%%%%%%%
%%%%%%%%%%%%%%%%%%%%%%%%%%%%%%%%%%%%%%%%%%%%%%%%%%%%%%%%%%%%%%%%%%%%%%%%%%%%%%%%%%%%
This paper is organized as follows: In Section 2, we state a proposition which is a necessary and sufficient condition for a linear map to be $k$-positive, provided it satisfies the equivariance property. In section 3, we give values of parameters $\alpha$ and $\beta$ for which the family of maps $\Phi_{\alpha, \beta,n}$ is positive. Later, we focus on the case of $\Phi_{\alpha, \beta,3}$ and study the properties of $2$-positivity, completely positivity and decomposability. We give values of parameters $\alpha$ and $\beta$ for which the family of maps $\Phi_{\alpha, \beta,3}$ is $2$-positive and not completely positive. By Proposition \ref{k-positivity}, the condition of being $(k-1)$-positive and not $k$-positive can easily be generalized for any natural numbers $n > 3$ and $2 < k \leq n.$

%%%%%%%%%%%%%%%%%%%%%%%%%%%%%%%%%%%%%%%%%%%%%%%%%%%%%%%%%%%%%%%%%%%%%%%%%%%%%%%%%%%%%%%%%%%%%%%%%%%%%%%%%%%%%%%%%%%%%%%%%%%%%%%%%%%%%%%%%%%%%%%%%%%%%%%%%%%%%%%%%%%%%%%%		
\section{Equivariance and Positivity} 
In this section, we give a criterion which is a necessary and sufficient condition for a linear map to be $k$-positive. We will show that equivariance is a useful property to conclude $k$-positivity of linear maps defined on finite-dimensional matrix algebras. We start with a particular case of positivity i.e. for $k=1$ and later generalize it for any $k \in \mathbb{N}.$ In this regard, maps with equivariance property are of special interest.

%%%%%%%%%%%%%%%%%%%%%%%%%%%%%%%%%%%%%%%%%%%%%%%%%%%%%%%%%%%%%%%%%%%%%%%%%%%%%%%%%%%%%%%%%%%%%%%%%%%%%%%%%%%%%%%%%%%%%%%%%%%%%%%%%%%%%%%%%%%%%%%%%%%%%%%%%%%%%%%%%%%%%%%%

 \subsection{$1$-positivity}
\begin{lemma} \label{mainlemma}
Let $\Phi:M_{d_{1}}(\mathbb C) \rightarrow M_{d_{2}}(\mathbb C)$ be a linear map such that
\begin{enumerate}
\item  \label{1}For every unitary matrix $\hspace{0.5mm}U \in M_{d_{1}}(\mathbb C),$ there exists a matrix $V \in M_{d_{2}}(\mathbb C)$(not \\ necessarily unitary) such that for every $X \in M_{d_{1}}(\mathbb C).$ we have, \begin{center}
	 $\Phi(UXU^{*})=V \Phi(X) V^{*}.$
\end{center} 
\item $\Phi(e_{11}) \geq 0,$ where $e_{11}$ is a matrix unit in $M_{d_{1}}(\mathbb C)$.
\end{enumerate}
Then $\Phi$ is positive.
\end{lemma}
 Property (\ref{1}) in above lemma is known as \emph{equivariance property}.
\begin{proof}
It is enough to prove that $\Phi(p) \in M_{d_{2}}(\mathbb C)^{+}$ for any rank one projection $p \in M_{d_{1}}(\mathbb C).$\\
We will use the fact that two hermitian matrices are unitarily equivalent if and only if they have the same eigenvalues (counting the multiplicity). Hence, $p$ and $e_{11}$ are unitarily equivalent.\\
Therefore, there exists a unitary matrix $U \in M_{d_{1}}(\mathbb C)$ such that $p = Ue_{11}U^{*}.$
\begin{align*}
\Phi(p)  &  =\Phi(Ue_{11}U^{*})\\
& =  V\Phi(e_{11})V^{*} \hspace{1mm} \text{(by equivariance property of} \hspace{1mm} \Phi). \\
\end{align*}
Given that $\Phi(e_{11}) \geq 0$ which implies $V\Phi(e_{11})V^{*} \geq 0.$\\
Therefore, $\Phi(p) \geq 0$ and $\Phi$ is positive. 
\end{proof}
%%%%%%%%%%%%%%%%%%%%%%%%%%%%%%%%%%%%%%%%%%%%%%%%%%%%%%%%%%%%%%%%%%%%%%%%%%%%%%%%%%%%%%%%%%%%%%%%%%%%%%%%%%%%%%%%%%%%%%%%%%%%%%%%%%%%%%%%%%%%%%%%%%%%%%%%%%%%%%%%%%%%%%%%
\subsection{$k$-positivity}
\begin{prop}\label{k-positivity}
Let $\Phi:{M}_{d_{1}}(\mathbb{C}) \rightarrow M_{d_{2}}(\mathbb{C})$ be a self adjoint linear map such that for every unitary matrix $U \in M_{d_{1}}(\mathbb C),$ there exists a matrix $V \in M_{d_{2}}(\mathbb C)$(not \\ necessarily unitary) such that for every $X \in M_{d_{1}}(\mathbb C).$ we have, \\ $$\Phi(UXU^{*})=V \Phi(X) V^{*}.$$ Then, for $k \leq min\{d_{1},d_{2}\},\hspace{0.5mm} \Phi$ is $k$-positive if and only if the block matrix $[\Phi(e_{ij})]_{i,j=1}^{k}$
is positive, where $(e_{ij})$ are matrix units in $M_{d_{1}}(\mathbb C).$ 
\end{prop}
\begin{proof}
Let $\Phi$ be $k$-positive. Then the map \\
$$i_{k} \otimes \Phi:M_{k}(\mathbb C) \otimes  M_{d_{1}}(\mathbb C)\rightarrow M_{k}(\mathbb C) \otimes  M_{d_{2}}(\mathbb C) $$ is positive.
Since the matrix $(e_{ij})_{i,j=1}^{k}$ is positive, where $(e_{ij})$ are matrix units in $M_{d_{1}}(\mathbb C)$. Hence,$$(i_{k} \otimes \Phi)(e_{ij})_{i,j=1}^{k}=[\Phi(e_{ij})]_{i,j=1}^{k}$$ is positive.
			
For the sufficient part, we prove that the map $i_{k} \otimes \Phi$ is positive. It is enough to prove that  $(i_{k} \otimes \Phi)p$ is positive for any rank-one projection $p \in M_{k}(\mathbb C) \otimes  M_{d_{1}}(\mathbb C).$ 
We prove it in two steps:\\
\vspace{2mm}
%\newpage
\underline{Step 1:}
$p=|x \rangle \langle x|$ for $|x\rangle \in \mathbb C^{k} \otimes \mathbb C^{d_{1}}, |||x \rangle || =1.$ \\
Consider the Schmidt decomposition of $|x\rangle$, we get 
\begin{align}
|x\rangle  & = (U \otimes V^{*}) \big(\sum_{i=1}^{k}c_{i} |e_{i} \rangle \otimes |e_{i}\rangle),\label{number1}
\end{align}

 where $c_{i}$'s are square root of eigenvalues of rank-one projection $|x \rangle \langle x|.$\\ Hence, $c_{i}$'s are positive real numbers with $\sum_{i=1}^{k}c_{i}^{2}=1.$\\ In expression (\ref{number1}), $U \in M_{k}(\mathbb C)$ and $V \in M_{d_{1}}(\mathbb C)$ are unitaries arising from \\ singular value decomposition of the matrix $C \in M_{k,d_{1}}(\mathbb C)$, where $C$ is the \\ coefficient matrix of $|x \rangle,$ when written in the canonical basis of  $\mathbb C^{k} \otimes \mathbb C^{d_{1}}$.
 \begin{align}\label{number2}
p & =|x \rangle \langle x| = (U \otimes V^{*})(\sum_{i,j=1}^{k}c_{i}c_{j}e_{ij} \otimes e_{ij})(U \otimes V^{*})^{*}.
 \end{align}

In expression (\ref{number2}), first component $e_{ij}$ of tensor product are matrix units in $M_{k}(\mathbb C)$ and second component are matrix units in $M_{d_{1}}(\mathbb C).$ \\  
$$p=(U \otimes V^{*})X(U \otimes V^{*})^{*}$$ where $X= \sum_{i,j=1}^{k}c_{i}c_{j}e_{ij}  \otimes e_{ij}.$ 
It can be easily seen that if $\Phi$ satisfy the equivariance property on $M_{d_{1}}(\mathbb C),$ then  $i_{k}\otimes \Phi$ also satisfy the equivariance property on $M_{k}(\mathbb C) \otimes M_{d_{1}}(\mathbb C),$ if unitaries are coming from $U_{k}(\mathbb C) \otimes U_{d_{1}}(\mathbb C)$, where $U_{m}(\mathbb C)$ denotes the set of all unitaries in $M_{m}(\mathbb C)$ for any natural number $m.$
			
By equivariance property of $i_{k}\otimes \Phi,$ there exist matrices $V_{1}$ and $V_{2}$ in $M_{k}(\mathbb C)$ and $M_{d_{2}}(\mathbb C)$ respectively such that 
\begin{align}
(i_{k}\otimes \Phi)p & =(i_{k}\otimes \Phi)[(U\otimes V^{*})X (U\otimes V^{*})^{*}]\\
& =(V_{1}\otimes V_{2})(i_{k}\otimes \Phi)X (V_{1}\otimes V_{2})^{*}.\label{number3} 
\end{align}
In fact $V_{1}=U.$\\
\underline{Step 2:}
In this step, we prove that $(i_{k}\otimes \Phi)p$ is positive if the block matrix $[\Phi(e_{ij})]_{i,j=1}^{k}$ is positive.
			
Let the block matrix $[\Phi(e_{ij})]_{i,j=1}^{k}$ be positive. First, we prove that $(i_{k}\otimes \Phi)X$ is positive.
Indeed,
\begin{align*}
(i_{k}\otimes \Phi)X & = (i_{k}\otimes \Phi) \bigg(\sum_{i,j=1}^{k}c_{i}c_{j}e_{ij}  \otimes e_{ij} \bigg) \\
& = [\Phi(c_{i}c_{j}e_{ij})]_{i,j=1}^{k}\\
& = [c_{i}c_{j}\Phi(e_{ij})]_{i,j=1}^{k}\\
& = (c_{i}c_{j})_{i,j=1}^{k} \circ [\Phi(e_{ij})]_{i,j=1}^{k}.
\end{align*}
			
The last term being the Schur product of two positive matrices is positive. Hence, $(i_{k}\otimes \Phi)X$ is positive.\\
By expression (\ref{number3}), $(i_{k}\otimes \Phi)p$ is positive if the block matrix $[\Phi(e_{ij})]_{i,j=1}^{k}$ is positive. \\
Therefore, we conclude that $\Phi$ is $k$-positive. 
\end{proof}
%%%%%%%%%%%%%%%%%%%%%%%%%%%%%%%%%%%%%%%%%%%%%%%%%%%%%%%%%%%%%%%%%%%%%%%%%%%%%%%%%%%%%%%%%%%%%%%%%%%%%%%%%%%%%%%%%%%%%%%%%%%%%%%%%%%%%%%%%%%%%%%%%%%%%%%%%%%%%%%%%%%%%%%%
\section{Properties of linear map $\Phi_{\alpha, \beta, n}$}
\subsection{Positivity}
%%%%%%%%%%%%%%%%%%%%%%%%%%%%%%%%%%%%%%%%%%%%%%%%%%%%%%%%%%%%%%%%%%%%%%%%%%%%%%%%%%%%
 \begin{prop} \label{Positivity}The linear map $\Phi_{\alpha, \beta,n}\hspace{1mm} (n \geq 3)$ is positive if and only if the following conditions hold.

\begin{enumerate}
\item $\alpha \in [0, \infty)$ if $\beta \geq 0.$
\item $\alpha \in \bigg[\frac{-(2+n\beta)+ \sqrt{n^{2}\beta^{2}-4(n-2)\beta+4}}{2} , \infty \bigg)$ if $\beta \leq 0.$
\end{enumerate} 
\end{prop}
%%%%%%%%%%%%%%%%%%%%%%%%%%%%%%%%%%%%%%%%%%%%%%%%%%%%%%%%%%%%%%%%%%%%%%%%%%%%%%%%%%%	
\begin{proof}
First, assume that $\Phi_{\alpha, \beta,n}$ is positive. By a calculation, It is easy to see that  $(\lambda- \alpha)^{n(n-2)}(\lambda-(1+\alpha))^{2(n-1)}\bigg[\lambda-\bigg(\alpha-\bigg(\frac{-(2+n\beta) \pm \sqrt{n^{2}\beta^{2}-4(n-2)\beta+4}}{2}\bigg)\bigg)\bigg]$ is the characteristic polynomial of $\Phi_{\alpha, \beta,n}(e_{11})$, which is positive. It is possible only when 
\begin{enumerate}
\item $\alpha \geq 0,$
\item $\alpha \geq \frac{-(2+n\beta)+ \sqrt{n^{2}\beta^{2}-4(n-2)\beta+4}}{2}.$
\end{enumerate}
We observe that above two conditions hold if and only if 
\[
 \begin{cases}
 \alpha \geq 0 &\quad\text{if} \hspace{2mm} \beta \geq 0, \\
 \alpha \geq \frac{-(2+n\beta)+ \sqrt{n^{2}\beta^{2}-4(n-2)\beta+4}}{2}  &\quad\text{if} \hspace{2mm} \beta \leq 0.
\end{cases}
\]
For the sufficient part, we apply Lemma \ref{mainlemma}.\\
Let $U \in M_{n}(\mathbb C)$ be a unitary matrix and $X \in M_{n}(\mathbb C)$ be any matrix. Then we have, 
$$\Phi_{\alpha, \beta,n}(UXU^{*})=\overline{U}X^{t}U^{t} \otimes \mathds{1}_{n}+\mathds{1}_{n}\otimes UXU^{*}+\mathrm{Tr}(X)(\alpha \mathds{1}_{n^{2}}+ \beta B_{n^{2}}).$$
\vspace{1mm}
Choose $V= \overline{U} \otimes U$ and use the fact that $B_{n^{2}}$ commutes with $\overline{U} \otimes U.$ We have,
\begin{align*}
V\Phi_{\alpha, \beta,n}(X)V^{*} &  =\overline{U}\otimes U(X^{t}\otimes \mathds{1}_{n}+\mathds{1}_{n}\otimes X+\mathrm{Tr}(X)(\alpha I_{n^{2}}+ \beta B_{n^{2}}))(\overline{U}\otimes U)^{*}\\
& = \overline{U}\otimes U (X^{t}\otimes \mathds{1}_{n}) U^{t}\otimes U^{*}+\overline{U}\otimes U (\mathds{1}_{n}\otimes X) U^{t}\otimes U^{*}+ \\& \hspace{4mm}  \mathrm{Tr}(X)(\overline{U}\otimes U)(\alpha \mathds{1}_{n^{2}}+ \beta B_{n^{2}})(U^{t} \otimes U^{*})\\
&= \overline{U}X^{t}U^{t} \otimes \mathds{1}_{n}+\mathds{1}_{n}\otimes UXU^{*}+\mathrm{Tr}(X)(\alpha \mathds{1}_{n^{2}}+ \beta B_{n^{2}})\\
&= \Phi_{\alpha, \beta,n}(UXU^{*}).
\end{align*}
\vspace{1mm}
Next, we need to show that the matrix $\Phi_{\alpha, \beta,n}(e_{11})$ is positive. It is clear from the reasoning given in the proof of necessary part, where values of parameters satisfy conditions (1) and (2) of the statement. 
Hence, $\Phi_{\alpha, \beta,n}$ is positive.
\end{proof}
Graphically, the region where $\Phi_{\alpha, \beta,3}$ is positive is shown below. The shaded region with blue colour represents the values of parameters $\alpha$ and $\beta$ where $\Phi_{\alpha, \beta,3}$ is positive.
%%%%%%%%%%%%%%%%%%%%%%%%%%%%%%%%%%%%%%%%%%%%%%%%%%%%%%%%%%%%%%%%%%%%%%%%%%%%%%%%%%%%%%%%%%%%%%%%%%%%%%%%%%%%%%%%%%%%%%%%%%%%%%%%%%%%%%%%%%%%%%%%%%%%%%%%%%%%%%%%%%%%%%%%
\pgfplotsset{compat=1.5}
\usepgfplotslibrary{fillbetween}
\usetikzlibrary{patterns}
\begin{figure}[h]
\centering
\begin{tikzpicture}
\begin{axis}[
xmin=-4, xmax=4,
ymin=-4, ymax=4,
axis lines=middle,
xlabel = $\beta$,
ylabel = {$\alpha$},
legend pos= outer north east
]  
\addplot[name path=F,blue,domain={-4:4}] {(-2-3*x+(9*x^2 - 4*x + 4)^0.5)*0.5} node[pos=4, above]{$f$};
\addlegendentry{$\alpha=\frac{-(2+3\beta)+\sqrt{9\beta^{2}-4\beta+4}}{2}$}

\addplot[name path=I,white,domain={-4:5}] {4}node[pos=.1, below]{$$};
\addplot[name path=H,white,domain={-4:5}] {-4}node[pos=.1, below]{$h$};
\addplot[name path=K,black,domain={-4:5}] {0}node[pos=.1, below]{$$};
\addplot[pattern=north west lines, pattern color=blue!50]fill between[of=F and I, soft clip={domain=-1.5:0}];
\addplot[pattern=north west lines, pattern color=blue!50]fill between[of=I and K, soft clip={domain=0:4}];
\end{axis}
\end{tikzpicture}
\caption{Region of Positivity}	
\end{figure}
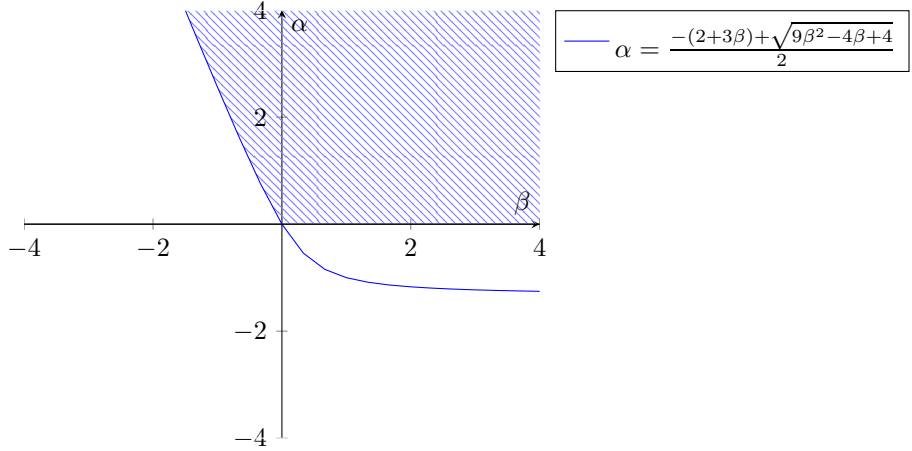
%%%%%%%%%%%%%%%%%%%%%%%%%%%%%%%%%%%%%%%%%%%%%%%%%%%%%%%%%%%%%%%%%%%%%%%%%%%%%%%%%%%%%%%%%%%%%%%%%%%%%%%%%%%%%%%%%%%%%%%%%%%%%%%%%%%%%%%%%%%%%%%%%%%%%%%%%%%%%%%%%%%%%%%%
\note
From now on, we focus on the case of $n=3.$ The family of maps $\Phi_{\alpha, \beta,3}$ is denoted as $\Phi[\alpha, \beta]$.
%%%%%%%%%%%%%%%%%%%%%%%%%%%%%%%%%%%%%%%%%%%%%%%%%%%%%%%%%%%%%%%%%%%%%%%%%%%%%%%%%%%%
\subsection{Completely Positivity}
%%%%%%%%%%%%%%%%%%%%%%%%%%%%%%%%%%%%%%%%%%%%%%%%%%%%%%%%%%%%%%%%%%%%%%%%%%%%%%%%%%%%
\begin{prop}\label{CP}
The map $\Phi[\alpha, \beta]$ is completely positive if and only if it is completely copositive if and only if the following conditions hold.\newpage
\begin{enumerate} 
\item \label{1-condition}$\alpha \in \bigg[\frac{-(3+3\beta)+\sqrt{9\beta^{2}-10\beta+17}}{2}, \infty\bigg)$ if $\beta \leq -0.2$.
\item \label{2-condition} $\alpha \in \bigg [\frac{-(3+3\beta)+\sqrt{9\beta^{2}-10\beta+17}}{2}, \infty \bigg) \cap [1, \infty)$ if $\beta \geq -0.2$.
\end{enumerate}
\end{prop}
%%%%%%%%%%%%%%%%%%%%%%%%%%%%%%%%%%%%%%%%%%%%%%%%%%%%%%%%%%%%%%%%%%%%%%%%%%%%%%%%%%%%
\begin{proof}
To prove the above statement, we work on the eigenvalues of the Choi matrix $C_{\Phi[\alpha, \beta]}$ of ${\Phi[\alpha, \beta]}$. By a calculation, the characteristic polynomial of $C_{\Phi[\alpha, \beta]}$ is given as: $(\lambda-(-1+\alpha))^{3}(\lambda-(1+\alpha))^{12}(\lambda-\alpha)^{6}\bigg[\lambda-\bigg(\alpha+\frac{(3+3\beta) \pm \sqrt{9\beta^{2}-10\beta+17}}{2}\bigg)\bigg]^{3}.$ \\
\vspace{0.5mm}
To prove the necessary part, assume that $\Phi[\alpha, \beta]$ is completely positive. Hence, all the eigenvalues of the Choi matrix $C_{\Phi[\alpha, \beta]}$ are positive. That is,
\begin{enumerate}
\item $\alpha \geq 1,$
\item	$\alpha \geq \frac{-(3+3\beta)+\sqrt{9\beta^{2}-10\beta+17}}{2}.$
\end{enumerate}
It is evident that if $\beta \leq -0.2,$ then $\alpha \geq \frac{-(3+3\beta)+\sqrt{9\beta^{2}-10\beta+17}}{2} \geq 1$	
Hence, it is enough to take, 
\begin{enumerate}
\item $\alpha \in \bigg[\frac{-(3+3\beta)+\sqrt{9\beta^{2}-10\beta+17}}{2} , \infty \bigg)$ when $\beta \leq -0.2$.
\item $\alpha \in \bigg[\frac{-(3+3\beta)+\sqrt{9\beta^{2}-10\beta+17}}{2}, \infty\bigg) \cap [1, \infty)$ when $\beta \geq -0.2$.
\end{enumerate}
For the sufficient part, let $\Phi[\alpha, \beta]$ satisfy conditions (1) and (2) of above Proposition. From the above reasoning, it can be easily seen that the Choi matrix of $\Phi[\alpha, \beta]$ is positive. Hence, $\Phi[\alpha, \beta]$ is completely positive. \\
It is easy to see that, $$(i_{3}\otimes T_{3} \otimes T_{3})C_{\Phi[\alpha, \beta]} =C_{\Phi[\alpha, \beta] \circ T_{3}}.$$

Therefore, $\Phi[\alpha, \beta]$ is completely positive if and only if it is completely copositive if and only if $\Phi[\alpha, \beta]$ satisfy conditions (\ref{1-condition}) and (\ref{2-condition}) of the statement.\\
\end{proof} 
%%%%%%%%%%%%%%%%%%%%%%%%%%%%%%%%%%%%%%%%%%%%%%%%%%%%%%%%%%%%%%%%%%%%%%%%%%%%%%%%%%%%
%%%%%%%%%%%%%%%%%%%%%%%%%%%%%%%%%%%%%%%%%%%%%%%%%%%%%%%%%%%%%%%%%%%%%%%%%%%%%%%%%%	
\begin{note}
	In Proposition \ref{CP}, we see that $\Phi[\alpha, \beta]$(In fact, $\Phi_{\alpha, \beta,n}$) is completely positive if and only if it is completely copositive. This is not the case for any linear map satisfying equivariance property. For example, consider the transpose map on finite-dimensional matrix algebras. Clearly, it satisfies equivariance property (Take $V=\overline{U}$). It is completely copositive and not completely positive. \\
	Hence, this property is not necessarily true for any map satisfying equivariance property. 
\end{note}
%%%%%%%%%%%%%%%%%%%%%%%%%%%%%%%%%%%%%%%%%%%%%%%%%%%%%%%%%%%%%%%%%%%%%%%%%%%%%%%%%%%%
%%%%%%%%%%%%%%%%%%%%%%%%%%%%%%%%%%%%%%%%%%%%%%%%%%%%%%%%%%%%%%%%%%%%%%%%%%%%%%%%%%%%
\begin{rem}
Graphically, the portion where $\Phi[\alpha, \beta]$ is positive and completely positive can be represented as follows. Dotted region in red colour describes the values of parameters $\alpha$ and $\beta$ for which the family of maps $\Phi[\alpha, \beta]$ is completely positive. 
\end{rem}
\newpage
%%%%%%%%%%%%%%%%%%%%%%%%%%%%%%%%%%%%%%%%%%%%%%%%%%%%%%%%%%%%%%%%%%%%%%%%%%%%%%%%%%%%%%%%%%%%%%%%%%%%%%%%%%%%%%%%%%%%%%%%%%%%%%%%%%%%%%%%%%%%%%%%%%%%%%%%%%%%%%%%%%%%%%% 
 \pgfplotsset{compat=1.10}
 \usepgfplotslibrary{fillbetween}
 \usetikzlibrary{patterns}
 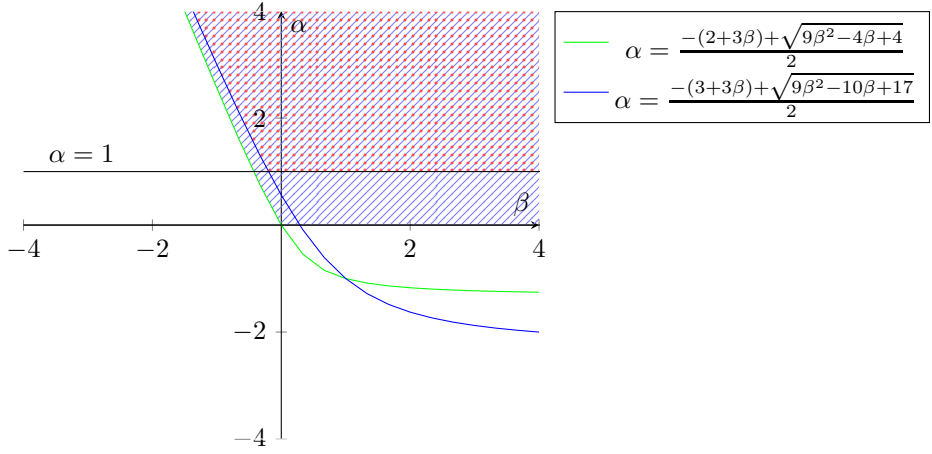
\begin{figure}[h]
 	\centering
 	\begin{tikzpicture}
 	\begin{axis}[
 	xmin=-4, xmax=4,
 	ymin=-4, ymax=4,
 	axis lines=middle,
 	xlabel = $\beta$,
 	ylabel = {$\alpha$},
 	legend pos= outer north east
 	]  
\addplot[name path=F,green,domain={-4:4}] {(-2-3*x+(9*x^2 - 4*x + 4)^0.5)*0.5} node[pos=4, above]{$f$};
\addlegendentry{$\alpha=\frac{-(2+3\beta)+\sqrt{9\beta^{2}-4\beta+4}}{2}$}
\addplot[name path=G,blue,domain={-4:4}] {(-3-3*x+(9*x^2 - 10*x + 17)^0.5)*0.5} node[pos=4, above]{$g$};
\addlegendentry{$\alpha=\frac{-(3+3\beta)+\sqrt{9\beta^{2}-10\beta+17}}{2}$}
\addplot[name path=I,white,domain={-4:5}] {4}node[pos=.1, below]{$$};
\addplot[name path=H,white,domain={-4:5}] {-4}node[pos=.1, below]{$h$};
\addplot[name path=L,black,domain={-4:5}] {0}node[pos=.1, below]{$$};
\addplot[name path=J,black,domain={-4:5}] {1}node[pos=.1, above]{$\alpha=1$};
\addplot[pattern=north east lines, pattern color=blue!50]fill between[of=F and I, soft clip={domain=-1.5:0}];
\addplot[pattern=north east lines, pattern color=blue!50]fill between[of=I and K, soft clip={domain=0:4}];
\addplot[pattern= dots, pattern color=red!70]fill between[of=G and I, soft clip={domain=-2:-0.2}];
\addplot[pattern= dots, pattern color=red!70]fill between[of=J and I, soft clip={domain=0:4}];
\addplot[pattern=dots, pattern color=red!70]fill between[of=I and J, soft clip={domain=-0.2:0}];
\end{axis}
\end{tikzpicture}
\caption{Region of Completely Positivity}
\end{figure}
%%%%%%%%%%%%%%%%%%%%%%%%%%%%%%%%%%%%%%%%%%%%%%%%%%%%%%%%%%%%%%%%%%%%%%%%%%%%%%%%%%%%%%%%%%%%%%%%%%%%%%%%%%%%%%%%%%%%%%%%%%%%%%%%%%%%%%%%%%%%%%%%%%%%%%%%%%%%%%%%%%%%%%%%

Propositions \ref{Positivity} and \ref{CP} give two different sets of values of parameters for which the family of maps  $\Phi[\alpha, \beta]$ is positive and completely positive. The next theorem gives the values of parameters $\alpha$ and $\beta$ for which the map $\Phi[\alpha, \beta]$ is positive and not completely positive. 

%%%%%%%%%%%%%%%%%%%%%%%%%%%%%%%%%%%%%%%%%%%%%%%%%%%%%%%%%%%%%%%%%%%%%%%%%%%%%%%%%%%%%%%%%%%%%%%%%%%%%%%%%%%%%%%%%%%%%%%%%%%%%%%%%%%%%%%%%%%%%%%%%%%%%%%%%%%%%%%%%%%%%%%%
\begin{theorem}\label{P and not CP}
The linear map $\Phi[\alpha, \beta]$ is positive and not completely positive if and only if 

\begin{enumerate}
\item $\alpha \in \bigg[\frac{-(2+3\beta)+\sqrt{9\beta^{2}-4\beta+4}}{2},\frac{-(3+3\beta)+\sqrt{9\beta^{2}-10\beta+17}}{2}\bigg)$ if $\beta \leq -0.2 $.
\item $\alpha \in \bigg[\frac{-(2+3\beta)+\sqrt{9\beta^{2}-4\beta+4}}{2},1\bigg)$ if $-0.2 \leq \beta \leq 0$.\\
\item $\alpha \in [{0,1})$ if $\beta \geq 0$.
\end{enumerate}
\end{theorem}
%%%%%%%%%%%%%%%%%%%%%%%%%%%%%%%%%%%%%%%%%%%%%%%%%%%%%%%%%%%%%%%%%%%%%%%%%%%%%%%%%%%%
\begin{proof}	
It is clear from Propositions \ref{Positivity} and \ref{CP}.
\end{proof}
\begin{rem}
Figure (\ref{graph3}) represents the portion where the map $\Phi[\alpha, \beta]$ is positive, completely positive and not completely positive. Starred region in green colour gives the values of parameters $\alpha$ and $\beta$ for which the map $\Phi[\alpha, \beta]$ is positive and not completely positive.
\end{rem} 	
%%%%%%%%%%%%%%%%%%%%%%%%%%%%%%%%%%%%%%%%%%%%%%%%%%%%%%%%%%%%%%%%%%%%%%%%%%%%%%%%%%%%%%%%%%%%%%%%%%%%%%%%%%%%%%%%%%%%%%%%%%%%%%%%%%%%%%%%%%%%%%%%%%%%%%%%%%%%%%%%%%%%%%%%
\pgfplotsset{compat=1.10}
\usepgfplotslibrary{fillbetween}
\usetikzlibrary{patterns}
\begin{figure}[h]
\centering
\begin{tikzpicture}
\begin{axis}[
xmin=-4, xmax=4,
ymin=-4, ymax=4,
axis lines=middle,
xlabel = $\beta$,
ylabel = {$\alpha$},
legend pos= outer north east
] 
\addplot[name path=F,green,domain={-4:4}] {(-2-3*x+(9*x^2 - 4*x + 4)^0.5)*0.5} node[pos=4, above]{$f$};
\addlegendentry{$\alpha=\frac{-(2+3\beta)+\sqrt{9\beta^{2}-4\beta+4}}{2}$} 
\addplot[name path=G,blue,domain={-4:4}] {(-3-3*x+(9*x^2 - 10*x + 17)^0.5)*0.5} node[pos=4, above]{$g$};
\addlegendentry{$\alpha=\frac{-(3+3\beta)+\sqrt{9\beta^{2}-10\beta+17}}{2}$}	
\addplot[name path=I,white,domain={-4:5}] {4}node[pos=.1, below]{$$};
\addplot[name path=H,white,domain={-4:5}] {-4}node[pos=.1, below]{$h$};
\addplot[name path=L,black,domain={-4:5}] {0}node[pos=.1, below]{$$};
\addplot[name path=J,black,domain={-4:5}] {1}node[pos=.1, above]{$\alpha=1$};
\addplot[pattern=north east lines, pattern color=blue!50]fill between[of=F and I, soft clip={domain=-1.5:0}];
\addplot[pattern=north east lines, pattern color=blue!50]fill between[of=I and K, soft clip={domain=0:4}];
\addplot[pattern= dots, pattern color=red!70]fill between[of=G and I, soft clip={domain=-2:-0.2}];
\addplot[pattern= dots, pattern color=red!70]fill between[of=J and I, soft clip={domain=0:4}];
\addplot[pattern= dots, pattern color=red!70]fill between[of=I and J, soft clip={domain=-0.2:0}];
\addplot[pattern= sixpointed stars, pattern color=green!70]fill between[of=F and G, soft clip={domain=-4:-0.2 }];
\addplot[pattern= sixpointed stars, pattern color=green!70]fill between[of=F and J, soft clip={domain=-0.2:0}];
\addplot[pattern= sixpointed stars, pattern color=green!70]fill between[of=L and J, soft clip={domain=0:4}];
\end{axis}
\end{tikzpicture}
\caption{Region of Positivity and not Completely positivity}
\label{graph3}
\end{figure}
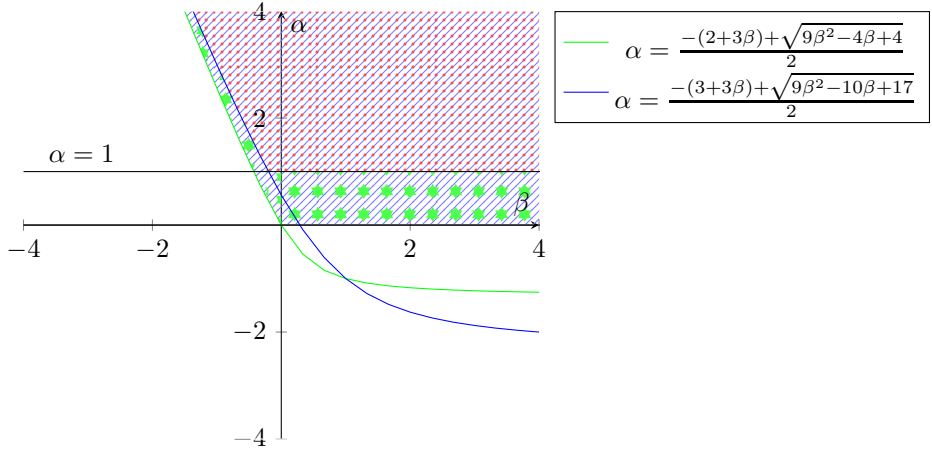	
%%%%%%%%%%%%%%%%%%%%%%%%%%%%%%%%%%%%%%%%%%%%%%%%%%%%%%%%%%%%%%%%%%%%%%%%%%%%%%%%%%%%%%%%%%%%%%%%%%%%%%%%%%%%%%%%%%%%%%%%%%%%%%%%%%%%%%%%%%%%%%%%%%%%%%%%%%%%%%%%%%%%%%%%
\subsection{$2$-positivity}
\normalfont{Choi \cite{choi1072positive} gave the first example of a linear map defined on $M_{3}(\mathbb{C})$, which is $2$-positive and not completely positive. We show that there exist values of parameters $\alpha$ and $\beta$ for which the family of maps $\Phi[\alpha, \beta]$ is $2$-positive and not completely positive.}
%%%%%%%%%%%%%%%%%%%%%%%%%%%%%%%%%%%%%%%%%%%%%%%%%%%%%%%%%%%%%%%%%%%%%%%%%%%%%%%%%%%% 
 \begin{theorem}\label{2-positivity}
 The linear map $\Phi[\alpha, \beta]$ is $2$-positive if and only if $\beta \in \mathbb{R}$ and $\alpha \in [1,\infty)\cap[-S_{\beta},\infty),$ where $S_{\beta}$ is the smallest root of the polynomial \\ $\lambda^{3}+(-2-3\beta)\lambda^{2}+(-2+4\beta)\lambda+2\beta.$
 \end{theorem}
 %%%%%%%%%%%%%%%%%%%%%%%%%%%%%%%%%%%%%%%%%%%%%%%%%%%%%%%%%%%%%%%%%%%%%%%%%%%%%%%%%%%
 
 \begin{proof} 
Let $\Phi[\alpha, \beta]$ be $2$-positive. By the definition of $2$-positive linear maps, $(i_{2}\otimes \Phi[\alpha, \beta])$ is positive. 
Hence, $(i_{2}\ \otimes\Phi[\alpha, \beta])[(e_{ij})]_{i,j=1}^{2} \geq 0$ where $(e_{ij})$ are matrix units in $M_{3}(\mathbb C).$ \\
By a calculation, the characteristic polynomial of  $[\Phi[\alpha, \beta](e_{ij})]_{i,j=1}^{2}$ is \\ $(\lambda -(2+\alpha))(\lambda- (-1+\alpha))(\lambda-(1+\alpha))^{7} (\lambda-\alpha)^{3}((\lambda-\alpha)^{3}+(-2-3\beta)(\lambda-\alpha)^{2}+(-2+4\beta)(\lambda-\alpha)+2\beta)^{2}$. \\ Therefore, we have the required bound on the values of parameters $\alpha$ and $\beta$.\\
To prove the sufficient part, we use Proposition \ref{k-positivity}. It is clear from the proof of Proposition \ref{Positivity} that $\Phi[\alpha, \beta]$ has equivariance property. Given the conditions on the values of parameters $\alpha$ and $\beta,$ it follows that the matrix $[\Phi[\alpha, \beta](e_{ij})]_{i,j=1}^{2}$ is positive.
Hence, the proof is completed.
\end{proof}
%%%%%%%%%%%%%%%%%%%%%%%%%%%%%%%%%%%%%%%%%%%%%%%%%%%%%%%%%%%%%%%%%%%%%%%%%%%%%%%%%%%%%%%%%%%%%%%%%%%%%%%%%%%%%%%%%%%%%%%%%%%%%%%%%%%%%%%%%%%%%%%%%%%%%%%%%%%%%%%%%%%%%%%%
\begin{theorem}\label{2-positive}
The linear map $\Phi[\alpha, \beta]$ is $2$-positive and not completely positive if and only if $\beta \leq -0.2$ and 
\begin{equation*}
\begin{split}
\alpha \in \bigg[\frac{-(2+3\beta)+\sqrt{9\beta^{2}-4\beta+4}}{2},\frac{-(3+3\beta)+\sqrt{9\beta^{2}-10\beta+17}}{2}\bigg) \\ 
\cap [1,\infty) \cap  [-S_{\beta},\infty),
\end{split}
\end{equation*}
where $S_{\beta}$ is the smallest root of the polynomial $\lambda^{3}+(-2-3\beta)\lambda^{2}+(-2+4\beta)\lambda+2\beta.$
%%%%%%%%%%%%%%%%%%%%%%%%%%%%%%%%%%%%%%%%%%%%%%%%%%%%%%%%%%%%%%%%%%%%%%%%%%%%%%%%%%%%
\begin{proof}
It is clear from Theorems \ref{P and not CP} and \ref{2-positivity} that the conditions for $\alpha$ and $\beta$ are sufficient for $\Phi[\alpha, \beta]$ to be $2$-positive and not completely positive. We need to consider those values of parameters $\alpha$ and $\beta$ for which the map $\Phi[\alpha, \beta]$ is $2$-positive and not completely positive. If we consider  $\beta > -0.2,$ then by Theorem \ref{P and not CP}, $\alpha < 1$  which destroys the $2$-positivity of linear map $\Phi[\alpha, \beta]$. \\
Hence, $\beta \leq -0.2$ and $\alpha \in \bigg[\frac{-(2+3\beta)+\sqrt{9\beta^{2}-4\beta+4}}{2},\frac{-(3+3\beta)+\sqrt{9\beta^{2}-10\beta+17}}{2}\bigg)$.\\ Therefore, we conclude that \\ $\alpha \in \bigg[\frac{-(2+3\beta)+\sqrt{9\beta^{2}-4\beta+4}}{2},\frac{-(3+3\beta)+\sqrt{9\beta^{2}-10\beta+17}}{2}\bigg) 
\cap [1,\infty) \cap [-S_{\beta},\infty) \hspace{1mm}\text{if} \hspace{1mm} \\ \beta 
\leq -0.2.$
\end{proof}
\end{theorem}
%%%%%%%%%%%%%%%%%%%%%%%%%%%%%%%%%%%%%%%%%%%%%%%%%%%%%%%%%%%%%%%%%%%%%%%%%%%%%%%%%%%%%%%%%%%%%%%%%%%%%%%%%%%%%%%%%%%%%%%%%%%%%%%%%%%%%%%%%%%%%%%%%%%%%%%%%%%%%%%%%%%%%%%%

\begin{rem}
The conditions of being $(k-1)$-positive and not $k$-positive can easily be generalized for any natural numbers $n > 3$ and $2 < k \leq n.$ From Proposition \ref{k-positivity}, it boils down to calculate the eigenvalues of block matrix $[\Phi[\alpha, \beta](e_{ij})]_{i,j=1}^{k}$ for $k$-positivity. Here, $(e_{ij})$ are matrix units in $M_{n}(\mathbb{C})$. 
\end{rem}

%%%%%%%%%%%%%%%%%%%%%%%%%%%%%%%%%%%%%%%%%%%%%%%%%%%%%%%%%%%%%%%%%%%%%%%%%%%%%%%%%%%%%%%%%%%%%%%%%%%%%%%%%%%%%%%%%%%%%%%%%%%%%%%%%%%%%%%%%%%%%%%%%%%%%%%%%%%%%%%%%%%%%%%%
\subsection{Decomposability}
\begin{defn}
A linear map is called decomposable if it can be written as a sum of a completely positive and a completely copositive map. Otherwise, it is called an indecomposable map.	
\end{defn} 
%%%%%%%%%%%%%%%%%%%%%%%%%%%%%%%%%%%%%%%%%%%%%%%%%%%%%%%%%%%%%%%%%%%%%%%%%%%%%%%%%%%%
\normalfont{In this section, we investigate the decomposability of the family of maps $\Phi[\alpha, \beta]$. Here we give necessary conditions of decomposability.}
%%%%%%%%%%%%%%%%%%%%%%%%%%%%%%%%%%%%%%%%%%%%%%%%%%%%%%%%%%%%%%%%%%%%%%%%%%%%%%%%%%%%%%%%%%%%%%%%%%%%%%%%%%%%%%%%%%%%%%%%%%%%%%%%%%%%%%%%%%%%%%%%%%%%%%%%%%%%%%%%%%%%%%%%
\begin{theorem}\label{decomposable new}
 The map $\Phi[\alpha, \beta]$ is decomposable if following conditions hold.
\begin{enumerate}
\item $\alpha \geq 0$ \hspace{1mm} if \hspace{1mm} $\beta \geq 0.$
\item  $\alpha \geq \frac{-(6+3\beta)+\sqrt{9\beta^{2}-28\beta+36}}{2}$ \hspace{1mm} if \hspace{1mm} $\beta \leq 0.$
\end{enumerate}	
\end{theorem}
\begin{proof}
Define $\Phi_{1}[\alpha,\beta]: M_{3}(\mathbb C) \rightarrow M_{3}(\mathbb C) \otimes M_{3}(\mathbb C)$ \hspace{0.5mm} given by \\ $$A \mapsto A^{t}\otimes \mathds{1}_{3}+\frac{\mathrm{Tr}(A)}{2}(\alpha \mathds{1}_{9}+\beta B_{9}).$$

And \\

$\Phi_{2}[\alpha,\beta]: M_{3}(\mathbb C) \rightarrow M_{3}(\mathbb C) \otimes M_{3}(\mathbb C)$ \hspace{0.5mm} given by \\ $$A \mapsto \mathds{1}_{3} \otimes A+\frac{\mathrm{Tr}(A)}{2} (\alpha \mathds{1}_{9}+\beta B_{9}).$$
By a calculation, $\Phi_{1}[\alpha,\beta]$ is completely copositive if and only if $\Phi_{2}[\alpha,\beta]$ is completely positive. The
characteristic polynomial of $C_{\Phi_{2}[\alpha,\beta]}$ is given by, $$\bigg(\lambda-\frac{\alpha}{2}\bigg)^{21}\bigg[\lambda-\bigg(\alpha+\bigg(\frac{6+3\beta \pm \sqrt{9\beta^{2}-28\beta+36}}{2}\bigg)\bigg)\bigg]^{3},$$
Therefore, $\Phi_{1}[\alpha,\beta]$ is completely copositive if and only if $\Phi_{2}[\alpha, \beta]$ is completely positive if and only if the following conditions hold.
\newpage
\begin{enumerate}
\item $\alpha \geq 0,$
\item $\alpha \geq \frac{-(6+3\beta)+\sqrt{9\beta^{2}-28\beta+36}}{2}.$
\end{enumerate}
Observe that $$\Phi[\alpha, \beta]= \Phi_{1}[\alpha,\beta]+\Phi_{2}[\alpha,\beta].$$
Hence, we conclude that $\Phi[\alpha, \beta]$ is decomposable if above two conditions hold.
\end{proof}
%%%%%%%%%%%%%%%%%%%%%%%%%%%%%%%%%%%%%%%%%%%%%%%%%%%%%%%%%%%%%%%%%%%%%%%%%%%%%%%%%%%%
%%%%%%%%%%%%%%%%%%%%%%%%%%%%%%%%%%%%%%%%%%%%%%%%%%%%%%%%%%%%%%%%%%%%%%%%%%%%%%%%%%%%
\pgfplotsset{compat=1.10}
\usepgfplotslibrary{fillbetween}
\usetikzlibrary{patterns}
\begin{figure}[h]
\centering
\begin{tikzpicture}
\begin{axis}
[
xmin=-4, xmax=4,
ymin=-4, ymax=4,
axis lines=middle,
xlabel = {$\beta$},
ylabel = {$\alpha$},
legend pos= outer north east]
\addplot[name path=F,green,domain={-4:4}] {(-2-3*x+(9*x^2 - 4*x + 4)^0.5)*0.5} node[pos=4, above]{$f$};
\addlegendentry{$\alpha=\frac{-(2+3\beta)+\sqrt{9\beta^{2}-4\beta+4}}{2}$}
\addplot[name path=Z,red,domain={-4:4}] {(-6-3*x+(9*x^2 - 28*x + 36)^0.5)*0.5} node[pos=4, above]{$x^2$};
\addlegendentry{$\alpha=\frac{-(6+3\beta)+\sqrt{9\beta^{2}-28\beta+36}}{2}$}
\addplot[name path=G,blue,domain={-4:4}]{(-3-3*x+(9*x^2 - 10*x + 17)^0.5)*0.5} node[pos=4, above]{x};
\addlegendentry{$\alpha=\frac{-(3+3\beta)+\sqrt{9\beta^{2}-10\beta+17}}{2}$}
\addplot[name path=L,black,domain={-4:5}] {0}node[pos=.1, below]{$$};
\addplot[name path=J,black,domain={-4:5}] {1}node[pos=.1, above]{$\alpha=1$};
\addplot[pattern= sixpointed stars, pattern color=black!70]fill between[of=F and G, soft clip={domain=-4:0}];
\addplot[pattern= sixpointed stars, pattern color=black!70]fill between[of=L and J, soft clip={domain=0:4}];
\addplot[pattern= vertical lines, pattern color=red!70]fill between[of=L and I, soft clip={domain=0:4}];
\addplot[pattern=vertical lines, pattern color=red!70]fill between[of=Z and I, soft clip={domain=-1.44:0}];
\end{axis}
\end{tikzpicture}
\caption{Region of  decomposability}
\label{graph4}
\end{figure}
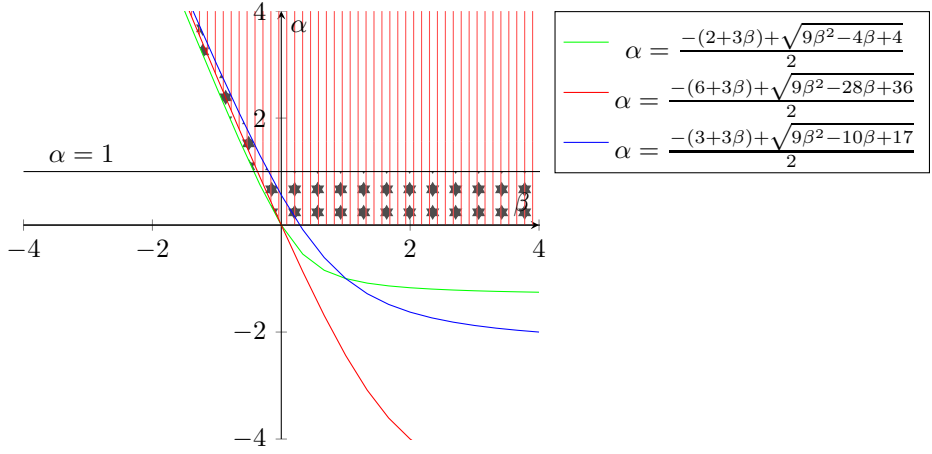
%%%%%%%%%%%%%%%%%%%%%%%%%%%%%%%%%%%%%%%%%%%%%%%%%%%%%%%%%%%%%%%%%%%%%%%%%%%%%%%%%%%%
%%%%%%%%%%%%%%%%%%%%%%%%%%%%%%%%%%%%%%%%%%%%%%%%%%%%%%%%%%%%%%%%%%%%%%%%%%%%%%%%%%% 
\normalfont{Figure(\ref{graph4}) gives values of parameters $\alpha$ and $\beta$ such that the map $\Phi[\alpha, \beta]$ is decomposable. The region of decomposability is shown by red vertical lines. It is evident from the figure(\ref{graph4}),
\[
 \begin{cases}
\alpha \geq \frac{-(3+3\beta)+\sqrt{9\beta^{2}-10\beta+17}}{2}  &\quad\text{if} \hspace{2mm} \beta \leq 0.\\
 \alpha \geq 1 &\quad\text{if} \hspace{2mm} \beta \geq 0. \\
 \end{cases}
\] 
For above values of parameters $\alpha$ and $\beta$, $\Phi[\alpha, \beta]$ is completely positive and hence decomposable. The region with vertical lines in red colour gives the values of those parameters.\\
The remaining portion is where $\Phi[\alpha, \beta]$ is decomposable and not completely positive, which is shown by black stars.}
\vspace{5mm}
%%%%%%%%%%%%%%%%%%%%%%%%%%%%%%%%%%%%%%%%%%%%%%%%%%%%%%%%%%%%%%%%%%%%%%%%%%%%%%%%%%% %%%%%%%%%%%%%%%%%%%%%%%%%%%%%%%%%%%%%%%%%%%%%%%%%%%%%%%%%%%%%%%%%%%%%%%%%%%%%%%%%%%%
\begin{corollary}\label{Yu}
The map $\Phi[\alpha, \beta]$ is decomposable and $2$-positive if 
\begin{equation*}
\begin{split}
\alpha \in \bigg[\frac{-(6+3\beta)+\sqrt{9\beta^{2}-28\beta+36}}{2} ,\frac{-(3+3\beta)+\sqrt{9\beta^{2}-10\beta+17}}{2}\bigg) \\
\cap [1,\infty) \cap [-S_{\beta},\infty)  \hspace{1.5mm}\text{and} \hspace{1.5mm} 	\beta \leq -0.2,\\
\end{split}
\end{equation*}
where $S_{\beta}$ is the smallest root of the polynomial $\lambda^{3}+(-2-3\beta)\lambda^{2}+(-2+4\beta)\lambda+2\beta$.
\end{corollary}
\begin{proof}
It is clear from Theorems \ref{2-positive} and \ref{decomposable new}.
\end{proof}
%%%%%%%%%%%%%%%%%%%%%%%%%%%%%%%%%%%%%%%%%%%%%%%%%%%%%%%%%%%%%%%%%%%%%%%%%%%%%%%%%%%%%%%%%%%%%%%%%%%%%%%%%%%%%%%%%%%%%%%%%%%%%%%%%%%%%%%%%%%%%%%%%%%%%%%%%%%%%%%%%%%%%%%%
\section{Conclusion}
\normalfont{In this paper, we have studied the properties of maps defined on finite-dimensional matrix algebras, which satisfy equivariance and showed that how equivariance can help to detect positivity and $k$-positivity of linear maps (Proposition \ref{k-positivity}). We have given a family of linear maps $\Phi_{\alpha, \beta,n}$ from $M_{n}(\mathbb{C})$ to $M_{n^{2}}(\mathbb{C})$ with two parameters $\alpha$ and $\beta$ which satisfy equivariance property and studied the properties of positivity, completely positivity, $2$-positivity and decomposability etc. Moreover, we determine the conditions of positivity of the family of linear maps $\Phi_{\alpha, \beta,n}$ for any natural number $n \geq 3$ and give the values of parameters $\alpha$ and $\beta$ for which the family of maps $\Phi_{\alpha, \beta,3}$ is $2$-positive and not completely positive. By using Proposition \ref{k-positivity}, the same calculation can be done for any natural number $n > 3.$ In this manner, we can obtain a class of linear maps which is $(k-1)$-positive and not $k$-positive for any natural numbers $n > 3$ and $2 < k \leq n.$ \\
It would be interesting to characterize the class of linear maps which satisfy equivariance property. By using Lemma \ref{mainlemma}, the characterization of linear maps with equivariance property can help to give a characterization of positive linear maps defined on finite-dimensional matrix algebras. \\
Yu Yang, Denny H. Leung, Wai-Shing Tang \cite{YANG} proved that every $2$-positive linear map from $M_{3}(\mathbb{C})$ to $M_{3}(\mathbb{C})$ is decomposable. They asked whether this observation is true for linear maps from  $M_{3}(\mathbb{C})$ to $M_{4}(\mathbb{C})$. Indeed, we give an affirmative answer to this observation for the family of linear maps $\Phi_{\alpha, \beta,3}$ (corollary \ref{Yu}), that is, there exist values of parameters $\alpha$ and $\beta$ for which the family of linear maps $\Phi_{\alpha, \beta,3}$ is decomposable and $2$-positive. \\
It would be interesting to inquire about the indecomposability of linear maps $\Phi_{\alpha, \beta,n}$. This criterion can be helpful to find PPT entangled states, (More details can be found in \cite{doi:10.1142/S0129055X13300021}).}
%%%%%%%%%%%%%%%%%%%%%%%%%%%%%%%%%%%%%%%%%%%%%%%%%%%%%%%%%%%%%%%%%%%%%%%%%%%%%%%%%%%%%%%%%%%%%%%%%%%%%%%%%%%%%%%%%%%%%%%%%%%%%%%%%%%%%%%%%%%%%%%%%%%%%%%%%%%%%%%%%%%%%%%%

\section{Acknowledgment}
\normalfont{The authors are indebted to Ivan Bardet for many useful discussions. GS would like to thank B.V.R Bhat for pointing out useful comments in corollary \ref{Yu}. HO was partially supported by KAKENHI Grant Number JP17K05285. BC's research was supported by ANR-14- CE25-0003. BC and GS were supported by JSPS Challenging Research grant 17K18734, JSPS Fund for the Promotion of Joint International Research grant 15KK0162, and JSPS wakate A grant 17H04823. GS would like to acknowledge FRIENDSHIP project of Japan International Corporation Agency (JICA) for research fellowship (D-15-90284).}
%%%%%%%%%%%%%%%%%%%%%%%%%%%%%%%%%%%%%%%%%%%%%%%%%%%%%%%%%%%%%%%%%%%%%%%%%%%%%%%%%%%%%%%%%%%%%%%%%%%%%%%%%%%%%%%%%%%%%%%%%%%%%%%%%%%%%%%%%%%%%%%%%%%%%%%%%%%%%%%%%%%%%%%%

\bibliographystyle{ieeetr}
\bibliography{report}

\noindent (Benoit Collins) Department of Mathematics, Graduate School of Science,
Kyoto University, Kyoto 606-8502, Japan \\
E-mail address: collins@math.kyoto-u.ac.jp\\

\noindent (Hiroyuki Osaka) Department of Mathematical Sciences. Ritsumeikan University, Kusatsu, Shiga 525-8577, Japan \\
E-mail address: osaka@se.ritsumei.ac.jp \\

\noindent (Gunjan Sapra) Department of Mathematics, Graduate School of Science,
Kyoto University, Kyoto 606-8502, Japan \\
E-mail address: gunjan18@math.kyoto-u.ac.jp

\end{document}